\documentclass[a4paper,11pt,fleqn]{article}
\usepackage{amsmath,amssymb,amsthm} 

\newcommand{\todo}[1][\null]{\ensuremath{\clubsuit}}

\newcommand{\noprint}[1]{}

\newcommand{\p}{\partial}

\newcommand{\const}{\mathop{\rm const}\nolimits}

\newcounter{tbn}

\newcounter{mcasenum}

\newtheorem{theorem}{Theorem}

\newtheorem{corollary}{Corollary}

\newtheorem*{proposition*}{Proposition}
{\theoremstyle{definition}

}

\begin{document}

\advance\topsep-5pt

\par\noindent {\Large\bf
Equivalence transformations in the study of integrability
}

{\vspace{5mm}\par\noindent\large Olena O.~Vaneeva$^{\dag 1}$, Roman O.~Popovych$^{\dag\ddag 2}$ and Christodoulos Sophocleous$^{\S 3}$
\par\vspace{5mm}\par}

{\par\noindent\it
${}^\dag$\ Institute of Mathematics of NAS of Ukraine, 3 Tereshchenkivska Str., Kyiv-4, 01601 Ukraine\\[.5ex]
$^\ddag$\;Wolfgang Pauli Institute, Nordbergstra{\ss}e 15, A-1090 Wien, Austria\\[.5ex]
$^\S$\;Department of Mathematics and Statistics, University of Cyprus, Nicosia CY 1678, Cyprus
}
\vspace{4mm}
{\par\noindent
$\phantom{{}^\dag{}\;}$E-mail: \it$^1$vaneeva@imath.kiev.ua,
$^2$rop@imath.kiev.ua,
$^3$christod@ucy.ac.cy
\par}

{\vspace{10mm}\par\noindent\hspace*{5mm}\parbox{150mm}{\small
We discuss how point transformations can be used for the study of integrability,
in particular, for deriving classes of integrable variable-coefficient differential equations.
The procedure of finding the equivalence groupoid of a class of differential equations is described
and then specified for the case of evolution equations.
A class of fifth-order variable-coefficient KdV-like equations is studied within the framework suggested.
}\par\vspace{4mm}}

\noprint{ Keywords: integrable models,
fifth-order Korteweg--de Vries equations, equivalence transformations, admissible transformations, normalized class  of differential equations}

\section{Introduction}

Since the 1960s exactly solvable (integrable) partial differential equations (PDEs) that model real-world phenomena have been a topic of permanent interest.
In particular, the inverse scattering transform method was introduced in~\cite{Gardner&Greene&Kruskal&Miura1967},
and was applied therein to the prominent Korteweg--de Vries (KdV) equation $u_t=u_{xxx}+6uu_x$~\cite{Korteweg&de_Vries1895} in order to find its soliton solutions.
The notion of soliton had appeared earlier in~\cite{Zabusky&Kruskal1965}.
It was shown in~\cite{Miura&Gardner&Kruskal1968} that the KdV equation possesses an infinite set of conservation laws of arbitrarily high orders,
and this property appeared to be typical for integrable equations.
A new direct method (the Hirota bilinear method) for finding multisoliton solutions to integrable nonlinear evolution equations was suggested in~\cite{Hirota1971}.
In contrast to the inverse scattering transform method, the Hirota bilinear method is algebraic rather than analytical.
Classical tools based on B\"acklund and Darboux transformations have been recalled after a long time of oblivion and intensively developed
\cite{Ablowitz&Segur1981,Matveev&Salle1991,Rogers&Schief2002,Sakhnovich2007}.
These and other methods were then applied to a wide range of integrable equations.
See, e.g., reviews of the results and other noteworthy references in~\cite{Ablowitz&Segur1981,Matveev&Salle1991,encycl}.

According to~\cite{Calogero1991}, integrable equations can be divided into those
that are linearizable by an appropriate {\it Change of variables} ($C$-integrable equations)
and equations integrable by the inverse scattering transform ({\it Spectral transform}) method ($S$-integrable equations).

Among $C$-integrable equations there are, e.g., the famous Burgers equation $u_t+ uu_x=\nu u_{xx}$ \cite{Burgers1948}
that can be linearized to the heat equation by the Hopf--Cole transformation~\cite{Cole1951,Hopf50},
the Sharma--Tasso--Olver equation $u_t+u_{xxx}+3u^2u_x+3u_x^2+3uu_{xx}=0$~\cite{Olver1977,Sharma&Tasso1977}, which is the second member of the Burgers hierarchy,
the $u^{-2}$-diffusion equation (also named the Fujita--Storm equation) $u_t=(u^{-2}u_x)_x+au$~\cite{Bluman&Kumei1980,Storm1951},
the Fokas--Yortsos equation $u_t=(u^{-2}u_x)_x+au^{-2}u_x$~\cite{Fokas&Yortsos1982,Strampp1982b}.
Further examples in (1+1)-dimensions can be found in~\cite{Calogero1991}.

$S$-integrable equations in (1+1)-dimensions include the KdV  and modified KdV equations,
the Gardner equation (the combined KdV--mKdV equation) $u_t+uu_x+u^2u_x+u_{xxx}=0$~\cite{Miura1968},
the cylindrical KdV equation $u_t=u_{xxx}+6uu_x-\frac{1}{2t}u$~\cite{Lugovtsov&Lugovtsov1969},
the Dym equation~$u_t=u^3u_{xxx}$~\cite{Kruskal1975},
the sine-Gordon equation $u_{tt}-u_{xx}+\sin u=0,$ etc. See other examples of integrable equations, e.g., in~\cite{encycl}.

Most of the integrable PDEs considered at the beginning of the development of integrability theory were constant-coefficient ones.
At the same time, many model equations appearing in applications explicitly involve independent variables.
For example, the generalized Burgers equations describing the propagation of weakly nonlinear acoustic waves
under the influence of geometrical spreading and thermoviscous diffusion in non-dimensional variables
are represented as $u_t+uu_x=g(t)u_{xx}$ with $ g\ne0$~\cite{HammertonCrighton1989},
and these equations are not $C$-integrable for nonconstant values of $g$.
The KdV and cubic Schr\"odinger equations with time-dependent coefficients,
\begin{gather}\label{vcKdVSchr}
u_t+f(t)uu_x+g(t)u_{xxx}=0\quad\mbox{and}\quad iu_t+f(t)u_{xx}+g(t)|u|^2u=0,
\end{gather}
respectively, also appear in various applications~\cite{Grimshaw1979a,Grimshaw1979b}.
Here $f$ and $g$ are nonvanishing smooth functions of $t$.

Many papers on variable-coefficient equations have been published, especially, in recent years.
The usual topics of these papers are
the application of the Painlev\'e test to single out subclasses
of integrable equations within wider classes of variable-coefficient equations,
the construction of conservation laws, Lax pairs and bilinear representations
and finding exact soliton solutions by the Hirota bilinear method.
Since variable-coefficient models are often quite complicated
and the number of variable coefficients varies from one to five, or even to ten in some cases,
packages of symbolic computations are widely used in such papers.
At the same time, the equivalence between equations within the classes under study was neglected in many works,
see a discussion in~\cite{Popovych&Vaneeva2010}.
On the other hand, even in pioneering works on exactly solvable models
it was shown that, if an integrable PDE is related to another PDE
by certain change of variables (point or non-point), then the latter PDE is also integrable.
The classical examples are the connection between the KdV and mKdV equation via the Miura transformation,
the reducibility of the Gardner equation to the mKdV~\cite{Miura1968}
and of the cylindrical KdV equation to the classical KdV~\cite{Johnson1979,Lugovtsov&Lugovtsov1969}.
Other examples are given in~\cite{Joshi1987}.
It was proved in the latter reference that the KdV and nonlinear Schr\"odinger equations
with time-dependent coefficients~\eqref{vcKdVSchr} pass the Painlev\'e test
if and only if the coefficients $f$ and $g$ satisfy the conditions $g(t)=f(t)(a_1\int^t\!f(s)\,{\rm d}s+a_0)$
and $g(t)=f(t)/(a_1\int^t\!f(s)\,{\rm d}s+a_0)$, respectively,
where $a_1$ and $a_0$ are constants with $a_1^2+a_0^2\neq0$.
These conditions coincide with those of reducibility of equations~\eqref{vcKdVSchr} to their constant-coefficient counterparts,
which were obtained in~\cite{Grimshaw1979a,Grimshaw1979b}.

Another way to construct variable-coefficient integrable models from constant-coefficient members of integrable hierarchies
was presented in~\cite[Theorem~3.1]{Fuchssteiner1993}.
Any ``linear superposition'', with arbitrary time-dependent coefficients, of members of an integrable evolution hierarchy
that correspond to mutually commuting flows proved to be also integrable.

The study of point transformations within a given class of variable-coefficient PDEs
and the derivation of conditions under which such equations reduce to constant-coefficient integrable equations allow one to obtain solutions,
conservation laws, other objects and related information in an easier way than using the direct computations for variable-coefficient equations.
The present paper is devoted to the discussion of this subject.
The consideration is illustrated by variable-coefficient fifth-order KdV-like equations.

\section{Admissible transformations in classes of differential equations}

Two differential equations connected by a certain change of variables (a point or contact transformation) are called {\it similar} ones~\cite{Ovsiannikov1982}.
Sets of objects (such as, e.g., classical solutions, conservation laws, symmetries, B\"acklund transformations, and, under certain restrictions for transformations, Lax pairs)
constructed for one of these equations can be used to derive the corresponding sets for the other equation.
This is why for a {\it class} of differential equations parameterized by arbitrary elements (constants or functions),
it is highly important to study the relations between fixed equations that are induced by point transformations.
Such similarity relations are called {\it allowed}~\cite{Winternitz92}, {\it form-preserving}~\cite{Kingston&Sophocleous1998}, and {\it admissible}~\cite{popo2010a} transformations.
Roughly speaking, an admissible transformation is a triple consisting of two fixed equations from a class
and a~point (resp.\ contact) transformation that links these equations.
The set of admissible  transformations of a class naturally possesses
the groupoid structure with respect to the composition of transformations
and, hence, it is called the \emph{equivalence groupoid} of the class~\cite{Popovych&Bihlo2012}.

More specifically, let~$\mathcal L_\theta$ be a differential equation $L(x,u_{(p)},\theta(x,u_{(p)}))=0$ for the~unknown function $u$
of $m$~independent variables $x=(x_1,\ldots,x_m).$
Here $u_{(p)}$ denotes the set of all the derivatives of~$u$ with respect to~$x$
of order not greater than~$p$, including $u$ as the derivative of order zero,
and $L$ is a fixed function depending on~$x,$ $u_{(p)}$ and~$\theta$.
Within the local approach, which is taken in the present paper, it is convenient to interpret
the equation $\mathcal L_{\theta}$ as an algebraic equation in the $p$-order jet space $J^{p}(x|u)$.%
\footnote{Roughly speaking, the $p$th order jet space $J^{p}(x|u)$ with underlying space of $(x,u)$ can be assumed
as the space whose coordinates represent the independent variables~$x$, the dependent variable~$u$
and the derivatives of~$u$ with respect to~$x$ up to order~$p$~\cite[p.~96]{Olver1993}.}
The tuple $\theta$ of $k$ arbitrary (parametric) functions $\theta^1(x,u_{(p)})$, \ldots, $\theta^k(x,u_{(p)})$
runs through the set~$\mathcal S$ of solutions of the auxiliary system $S(x,u_{(p)},\theta_{(q)}(x,u_{(p)}))=0$
of differential equations with respect to $\theta$.
In this system $x$ and $u_{(p)}$ play the role of independent variables
and $\theta_{(q)}$ stands for the set of all the partial derivatives of $\theta$ of order not greater than $q$
with respect to the variables $x$ and $u_{(p)}$.
Usually the set $\mathcal S$ is additionally constrained by the condition
$\Sigma(x,u_{(p)},\theta_{(q)}(x,u_{(p)}))\ne0$ with another differential function~$\Sigma$.
In what follows, we call the functions $\theta$ arbitrary elements,
and \emph{the class of equations~$\mathcal{L}_\theta$ with the arbitrary elements $\theta$ running through $\mathcal S$}
is denoted by~$\mathcal L|_{\mathcal S}$.

To illustrate the above notions, consider the class of $n$th order (1+1)-dimensional evolution equations,
\begin{equation}\label{EqGenEvol}
u_t=H(t,x,u_0,u_1,\dots,u_n),
\end{equation}
where  $n\geqslant 2$, $u_j\equiv \p^j u/\p x^j$, $j=1,2,\dots$, and $u_0\equiv u$.
We shall also employ, where convenient or necessary, the following notation for low-order derivatives:
$u_x=u_1$, $u_{xx}=u_2$, $u_{xxx}=u_3$, etc.
In general, a subscript of a function denotes the differentiation with respect to the corresponding variable,
e.g., $u_t\equiv \p u/\p t$, $H_{u_i}\equiv\p H/\p u_i$.
For the above class, the tuple of arbitrary elements $\theta$ consists of a single arbitrary smooth function $H$ of its arguments.
The auxiliary equations to~$H$ singling out evolution equations among all $n$th order two-dimensional partial differential equations
form the system
\[
H_{u_{it}}=0,\ i=0,\dots,n-1,\quad H_{u_{itt}}=0,\ i=0,\dots,n-2,\quad \dots,
\]
meaning that the arbitrary element~$H$ does not depend on derivatives of~$u$ involving the differentiation with respect to~$t$.
The condition that the equation order equals~$n$ leads to the auxiliary inequality $H_{u_n}\neq 0$.
For quasilinear evolution equations, the arbitrary element~$H$ is linear in the highest-order derivative~$u_n$, i.e.,
the subclass~$\mathcal E_{\rm ql}$ of such equations is singled out from the entire class~\eqref{EqGenEvol} by the additional auxiliary equation $H_{u_nu_n}=0$.
Representing~$H$ in the form $H=Fu_n+G$ and interpreting $F=F(t,x,u_0,u_1,\dots,u_{n-1})$ and $G=G(t,x,u_0,u_1,\dots,u_{n-1})$ as new arbitrary elements,
we reparameterize the subclass~$\mathcal E_{\rm ql}$.
In terms of~$F$ and~$G$, the system of auxiliary equations and inequality for the subclass~$\mathcal E_{\rm ql}$ takes the form
\begin{gather}\label{EqAuxilaryForQuasilinEvolEqs}
\begin{split}
&F_{u_{it}}=G_{u_{it}}=0,\ i=0,\dots,n-1,\quad F_{u_{itt}}=G_{u_{itt}}=0,\ i=0,\dots,n-2,\quad \dots, \\
&F_{u_n}=G_{u_n}=0,\quad F\ne0.
\end{split}
\end{gather}
Imposing additional auxiliary equations on~$H$ (resp.\ $F$ and~$G$), one can construct a hierarchy of nested subclasses of evolution equations.

An admissible transformation in the class~$\mathcal L|_{\mathcal S}$ is a triple $(\theta,\tilde\theta,\varphi)$,
where $\theta,\tilde\theta\in\mathcal S$ such that the equations $\mathcal L_\theta$ and $\mathcal L_{\tilde\theta}$ are similar,
and $\varphi$ is a point transformation from $\mathcal L_\theta$ to $\mathcal L_{\tilde\theta}$.
The set~$\mathcal G^\sim=\mathcal G^\sim(\mathcal L|_{\mathcal S})$ of admissible transformations
is naturally equipped with the groupoid structure and called the \emph{equivalence groupoid} of the class~$\mathcal L|_{\mathcal S}$.

The transformations acting in the joint space of variables and arbitrary elements and preserving the class~$\mathcal L|_{\mathcal S}$
form the \emph{equivalence group} $G^{\sim}$ of this class.
There exist several kinds of equivalence groups depending on restrictions that are imposed on the transformations.
The \emph{usual equivalence group} of the class~$\mathcal L|_{\mathcal S}$, roughly speaking,
consists of the point transformations in the space of variables and arbitrary elements,
which preserve the whole class~$\mathcal L|_{\mathcal S}$ and are projectable on the variable space,
i.e., the transformation components corresponding to independent and dependent variables do not depend on arbitrary elements.
If the transformations of variables depend on arbitrary elements $\theta$,
then the corresponding equivalence group is called the {\it generalized equivalence group}.
If new arbitrary elements are expressed via old ones in a nonlocal way (e.g., new arbitrary elements are expressed via integrals of old ones),
then the equivalence group is said to be {\it extended}.
\emph{Generalized extended equivalence groups} possess both of these properties.
In case of the single dependent variable~$u$, one can consider transformations that are contact rather than point with respect to the equation variables;
this leads to the notion of a \emph{contact equivalence group}.
A~number of examples of finding and utilizing different kinds of equivalence groups are presented, e.g., in~\cite{ips2010a,vps2009,vane2012b}.

If any admissible transformation in the class~$\mathcal L|_{\mathcal S}$ is induced by a transformation from its equivalence group~$G^{\sim}$
(usual / generalized / extended / generalized extended), then this class is called \emph{normalized} (in the corresponding sense).
In other words, the equivalence groupoid of a~normalized class is generated by the equivalence group of this class.

In order to find the (point) equivalence groupoid $\mathcal G^\sim$ of the class $\mathcal L|_{\mathcal S}$, the direct method is usually applied.
We fix two arbitrary equations,
\[
\mathcal L_\theta\colon L(x,u_{(p)},\theta(x,u_{(p)}))=0\quad\mbox{and}\quad
\mathcal L_{\tilde\theta}\colon  L(\tilde x,\tilde u_{(p)},\tilde \theta(\tilde x,\tilde u_{(p)}))=0,\]
from the class $\mathcal L|_{\mathcal S}$ and suppose that they are connected by a nondegenerate point transformation~$\mathcal T$
of the general form $\tilde x_i=X^i(x,u)$, $i=1,\dots,m,$ $\tilde u=U(x,u)$.
Changing the variables in $\mathcal L_{\tilde\theta}$,
we express all involved derivatives $\tilde u_{(p)}$ in terms of $u_{(p)}$ and partial derivatives of the functions $X^i$ and $U$
with respect to $x$ and $u$ and substitute them into $\mathcal L_{\tilde\theta}$, which gives an equation $\tilde{\mathcal L}$ for the function $u$.
Then  $\tilde{\mathcal L}$ has to be satisfied identically for solutions of $\mathcal L_{\theta}$.
This requirement leads to the system of determining equations for the components of~$\mathcal T$.
The procedure of deriving the determining equations is algorithmic and consists of the following steps:
We solve the equation $\mathcal L_{\theta}$ as an algebraic equation in the $p$-order jet space $J^{p}(x|u)$
with respect to a derivative (called the \emph{principal} derivative)
and substitute the derived expression involving other derivatives of~$u$ (called \emph{parametric} derivatives) into $\tilde{\mathcal L}$.
As a result, we get a complicated equation involving the variables~$x$,
parametric derivatives of $u$, derivatives of the transformation components $X^i$ and $U$,
and derivatives of the arbitrary elements $\theta$ and $\tilde\theta$,
which should be identically satisfied with respect to~$x$ and parametric derivatives of $u$ for each appropriate transformation~$\mathcal T$.
Applying different techniques, in particular, the differentiation and splitting of this identity with respect to~$x$ and parametric derivatives of~$u$
gives rise to the determining equations for the functions~$X^i$ and~$U$ that define the transformation~$\mathcal T$.

The equivalence groupoids of many classes of differential equations have already been constructed in the literature,
see \cite{Bihlo&Cardoso-Bihlo&Popovych2012,Kingston&Sophocleous1998,Popovych&Bihlo2012,popo2010a,Popovych&Vaneeva2010,vps2009,vane2012b,Winternitz92} and references therein.
If the class $\mathcal L|_{\mathcal S}$ is contained in a superclass $\bar{\mathcal L}|_{\bar{\mathcal S}}$ whose equivalence groupoid~$\bar{\mathcal G}^\sim$ is known,
then the equivalence groupoid~$\mathcal G^\sim$ of~$\mathcal L|_{\mathcal S}$ is a subgroupoid of~$\bar{\mathcal G}^\sim$,
and hence all the determining equations and conditions for elements of~$\bar{\mathcal G}^\sim$ can be taken into account
in the very beginning of the application of the direct method for finding~$\mathcal G^\sim$.
If a class of differential equations is normalized,
the construction of its equivalence group
leads to an exhaustive and concise description of its equivalence groupoid.
This is why it is especially convenient to describe equivalence groupoids for hierarchies of nested normalized classes.

We consider a chain of nested normalized classes of evolution equations, which is of interest in view of the subject of the present paper.
It is a well-known folklore assertion~\cite{Magadeev1993}
that any contact transformation~$\mathcal T$ relating two fixed equations $u_t=H$ and $\tilde u_{\tilde t}=\tilde H$ from the class~\eqref{EqGenEvol} has the form
$\tilde t=T(t)$, $\tilde x=X(t,x,u,u_x)$, $\tilde u=U(t,x,u,u_x)$.
In comparison with the general contact transformation in the space of $(t,x,u)$,
the peculiarity is that the transformation component for~$t$ depends only on~$t$
and the transformation component for all the variables does not depend on~$u_t$.
The contact and nondegeneracy assumptions are reduced for~$\mathcal T$ to the conditions
$(U_x+U_uu_x)X_{u_x}=(X_x+X_uu_x)U_{u_x}$ and $T_t\ne0$, $\mathop{\rm rank}\p(X,U)/\p(x,u,u_x)=2$, respectively.
The standard prolongation of~$\mathcal T$ to the derivatives $u_1$, \dots, $u_n$ is carried out using the chain rule,
which gives $\tilde u_{\tilde x}=V(t,x,u,u_x)$,
where $V=(U_x+U_uu_x)/(X_x+X_uu_x)$ or $V=U_{u_x}/X_{u_x}$ if $X_x+X_uu_x\ne0$ or $X_{u_x}\ne0$, respectively,
and $\tilde u_i\equiv\p^i \tilde u/\p \tilde x^i=((1/D_xX)D_x)^{i-1}V$, $i=2,\dots,n$.
Here $D_x=\p_x+u_x\p_u+u_{tx}\p_{u_t}+u_{xx}\p_{u_x}+\cdots$ is the operator of total differentiation with respect to the variable~$x$.
The possibility of vanishing $X_x+X_uu_x$ and $X_{u_x}$ simultaneously is ruled out by the nondegeneracy assumption.
Moreover, the contact and nondegeneracy assumptions jointly imply that $(X_u,U_u)\ne(0,0)$.
The transformed arbitrary element~$\tilde H$ is equal~to
\[
\tilde H=\frac{U_u-X_uV}{T_t}H+\frac{U_t-X_tV}{T_t}.
\]
Each of the above contact transformations maps the entire class~\eqref{EqGenEvol} onto itself.
Therefore, its prolongation to the arbitrary element~$H$ belongs to
the contact equivalence group~$G^\sim_{\rm c}$ of the class~\eqref{EqGenEvol},
and any element of~$G^\sim_{\rm c}$ can be obtained in this way.
In other words, the equivalence group~$G^\sim_{\rm c}$ generates the whole contact equivalence groupoid~$\mathcal G^\sim_{\rm c}$
of the class~\eqref{EqGenEvol}, i.e., this class is contact-normalized,
which obviously implies its normalization with respect to point transformations as well.
If $n\geqslant3$, the subclass~$\mathcal E_{\rm ql}$ of $n$th order quasilinear evolution equations has
the same contact equivalence group and is contact-normalized too.

Each next subclass is singled out from the previous one by sequentially adding more equations to the system~\eqref{EqAuxilaryForQuasilinEvolEqs}
while preserving the property of normalization in the usual sense.
The additional constraints $F_{u_2}=\dots=F_{u_{n-1}}=0$ (i.e., $F=F(t,x,u,u_x)$\,) lead to the principal narrowing of the equivalence groupoid
of the corresponding subclass: Its contact equivalence group coincides with its point equivalence group.
Thus, when considering equations from this subclass, it suffices to use only the point equivalence.
Imposing the additional constraint $F_{u_1}=0$, we obtain a subclass in which the $x$-component of any equivalence transformation does not involve~$u$, $X_u=0$,
i.e., all equivalence transformations are fiber-preserving.
The equivalence group of the subclass of equations with $F$ depending only on~$t$, $F=F(t)$, consists of transformations satisfying the equation~$X_{xx}=0$.
Finally, for the subclass of equations of the~form
\begin{gather}\label{EqvcKdVlikeSuperclass}
u_t=F(t)u_n+G(t,x,u_0,u_1,\dots,u_{n-1}),\quad F\ne0,\quad G_{u_iu_{n-1}}=0,\ i=1,\dots,n-1,
\end{gather}
where $n\geqslant 2$, any equivalence transformation is linear in~$u$ since $U_{uu}=0$.
Collecting all determining equations for admissible transformations in the class~\eqref{EqvcKdVlikeSuperclass},
which are exhausted by the above equations $X_u=X_{xx}=U_{uu}=0$,
we derive the following assertion.

\begin{theorem}
The usual point equivalence group of the class~\eqref{EqvcKdVlikeSuperclass} consists of the transformations
\begin{gather}\label{EqEquivtransOfvcKdVlikeSuperclass}
\tilde t=T(t),\quad \tilde x=X^1(t)x+X^0(t),\quad \tilde u=U^1(t,x)u+U^0(t,x), \quad
\tilde F=\frac{(X^1)^n}{T_t}F,
\\
\nonumber\tilde G=\frac{U^1}{T_t}G-\left(\sum^{n-1}_{k=0}\!\begin{pmatrix}n\\k\end{pmatrix}\!U^1_{n-k}u_k+U^0_n\!\right)\!\frac F{T_t}
+\frac{U^1_t}{T_t}u+\frac{U^0_t}{T_t}-\frac{X^1_tx+X^0_t}{T_tX^1}(U^1u_x+U^1_xu+U^0_x),
\end{gather}
where $T_tX^1U^1\neq0$,
and the class~\eqref{EqvcKdVlikeSuperclass} is normalized with respect to this group.
\end{theorem}

\section{Integrable subclasses in a class of fifth-order\\ variable-coefficient KdV-like equations}\label{SectionOnfKdV}

Consider the class of variable-coefficient fifth-order KdV-like equations of the form
\begin{gather}\arraycolsep=0ex\label{Eqvc5KdV0}
\begin{array}{l}
u_t+a(t)uu_{xxx}+b(t)u_xu_{xx}+c(t)u^2u_x+f(t)u u_x+g(t)u_{xxxxx}\\[1.5ex]\qquad {}+h(t)u_{xxx}+m(t)u+n(t)u_x+k(t)xu_x=0,
\end{array}
\end{gather}
where the functions $a$, $b$, $c$, $f$, $g$, $h$, $m$, $n$, and $k$ are arbitrary smooth functions of the time variable~$t$ with $g(a^2+b^2+c^2)\neq0$.
Recently certain subclasses of this class were studied, e.g., in~\cite{Xu2013,Yu&Gao&Sun&Liu2010,Yu&Gao&Sun&Liu2012}.

Thus, in~\cite{Xu2013} the integrability of equations from the class~\eqref{Eqvc5KdV0} with $k=0$ (and the re-denoted coefficients $f=d$, $g=l$ and $h=e$)
was investigated using the Painlev\'e test.
It was found that such equations are Painlev\'e integrable in the following three cases
\begin{gather*}
\mbox{\bf \phantom{II}I.}\quad b =a ,\quad c = \mu_1a e^{\int\!m\,{\rm d}t},\quad f =2\mu_2a ,\quad g =\frac{a}{5\mu_1} e^{-\int\!m\,{\rm d}t},\quad
h =\frac{\mu_2}{\mu_1} a e^{-\int\!m\,{\rm d}t} ;\\[1ex]
\mbox{\bf \phantom{I}II.}\quad b =2a ,\quad c = \mu_1a e^{\int\!m\,{\rm d}t},\quad f =2\mu_1h e^{\int\!m\,{\rm d}t},\quad g =\frac{3 a}{10\mu_1} e^{-\int\!m\,{\rm d}t};\\[1ex]
\mbox{\bf III.}\quad b =\frac52a ,\quad c = \mu_1a e^{\int\!m\,{\rm d}t},\quad  f =2\mu_2a ,\quad g =\frac{a}{5\mu_1} e^{-\int\!m\,{\rm d}t},\quad
h =\frac{\mu_2}{\mu_1} a e^{-\int\!m\,{\rm d}t}.
\end{gather*}
In all three cases $\mu_1$ and $\mu_2$ are arbitrary constants with $\mu_1\neq0$,
the functions $a$, $m$ and $n$ are arbitrary.
In Case~II the function $h$ is also arbitrary.
Here and in what follows an integral with respect to~$t$ should be interpreted as a fixed antiderivative of the integrand.
$N$-soliton solutions were constructed for the first two cases
whereas only one- and two-soliton solutions were presented in Case~III.

The same subclass of equations with $k=0$ was treated earlier in~\cite{Yu&Gao&Sun&Liu2010}.
Although it was stated that both the Painlev\'e test and the mapping to the completely integrable constant-coefficient counterparts were applied for separating integrable cases,
$N$-soliton solutions, a B\"acklund transformation and a Lax pair were constructed therein only for equations with additional constraints
\[
a=b=15g\nu e^{\int\!m\,{\rm d}t},\quad c=45g\nu^2e^{\,2\!\int\!m\,{\rm d}t},\quad f=h=0,
\]
where $\nu$ is a nonzero constant ($\nu=1/\rho$ in the notation of~\cite{Yu&Gao&Sun&Liu2010}),
which gives a particular subcase of Case~I.
The other integrable cases were missed.

In~\cite{Yu&Gao&Sun&Liu2012}, such objects were constructed
for equations of the form~\eqref{Eqvc5KdV0} with $f=h=0$ (and the re-denoted coefficients $g=d$ and $k=l$) under the constraints
\[
a=b=15g\nu e^{\int(m-2k)\,{\rm d}t},\qquad c=45g\nu^2e^{\,2\!\int(m-2k)\,{\rm d}t}.
\]
It was also indicated that these constraints are derived
both by Painlev\'e analysis and by mapping the corresponding variable-coefficient models to their completely integrable constant-coefficient counterparts.
In fact, this consideration just extends the results of~\cite{Yu&Gao&Sun&Liu2010} to the case of nonzero~$k$,
although the parameter~$k$ is not essential and can be set to zero by a point transformation.

We show that all the mentioned cases of integrable equations from the class~\eqref{Eqvc5KdV0} are reduced by point transformations
to well-known fifth-order integrable evolution equations.
To achieve this, we present a complete description of admissible transformations between equations from this class.

\begin{theorem}
The generalized extended equivalence group $G^\sim$ of the class~\eqref{Eqvc5KdV0} consists of  the transformations
\begin{gather}\label{Eqvc5KdVEquivGroup}
\begin{split}&
\tilde t=\alpha(t),\quad
\tilde x=\beta(t)x+\gamma(t),\quad
\tilde u=\varphi(t)\left(u+\sigma e^{-\int\!m\,{\rm d}t}\right),
\\&
\tilde a=\frac{\beta^3}{\alpha_t\varphi}a, \quad
\tilde b=\frac{\beta^3}{\alpha_t\varphi}b,\quad
\tilde c=\frac{\beta}{\alpha_t\varphi^2}c,\quad
\tilde f=\frac{\beta}{\alpha_t\varphi}\left(f-2\sigma ce^{-\int\!m\,{\rm d}t}\right),
\\&
\tilde g=\frac{\beta^5}{\alpha_t}g,\quad
\tilde h=\frac{\beta^3}{\alpha_t}\left(h-\sigma ae^{-\int\!m\,{\rm d}t}\right),\quad \tilde {m}=\frac{1}{\alpha_t}\left(m-\frac{\varphi_t}{\varphi}\right),
\\&
\tilde n=\frac{\beta}{\alpha_t}\left(n+\left(\frac{\gamma}\beta\right)_t-k\frac{\gamma}\beta+\sigma^2ce^{-2\int\!m\,{\rm d}t}-\sigma fe^{-\int\!m\,{\rm d}t}\right),\quad
\tilde k=\frac1{\alpha_t}\left(k+\frac{\beta_t}\beta\right),
\end{split}
\end{gather}
where $\alpha$, $\beta$, $\gamma$, and $\varphi$ run through the set of smooth functions of~$t$ with $\alpha_t\beta\varphi\ne0$, and $\sigma$ is an  arbitrary constant.
This group generates the entire equivalence groupoid $\mathcal G^\sim$ of the class~\eqref{Eqvc5KdV0},
i.e., the class~\eqref{Eqvc5KdV0} is normalized in the generalized extended sense.
\end{theorem}

\begin{proof}
Suppose that two arbitrary fixed equations from the class under consideration,
i.e., an equation of the form~\eqref{Eqvc5KdV0} and the equation
\begin{gather}\arraycolsep=0ex \label{Eqvc5KdV0tilde}
\begin{array}{l}
\tilde u_{\tilde t}+\tilde a(\tilde t)\tilde u\tilde u_{\tilde x\tilde x\tilde x}+\tilde b(\tilde t)\tilde u_{\tilde x}\tilde u_{\tilde x\tilde x}
+\tilde c(\tilde t)\tilde u^2\tilde u_{\tilde x}+\tilde f(\tilde t)\tilde u\tilde u_{\tilde x}+\tilde g(\tilde t)\tilde u_{\tilde x\tilde x\tilde x\tilde x\tilde x}
\\[1.5ex]\qquad\quad{}
+\tilde h(\tilde t)\tilde u_{\tilde x\tilde x\tilde x}+\tilde m(\tilde t)\tilde u+\tilde n(\tilde t)\tilde u_{\tilde x}+\tilde k(\tilde t)\tilde x\tilde u_{\tilde x}=0
\end{array}
\end{gather}
are similar with respect to a contact transformation~$\mathcal T$.
As the class~\eqref{Eqvc5KdV0} is a subclass of the normalized class~\eqref{EqvcKdVlikeSuperclass},
the transformation~$\mathcal T$ satisfies all the determining equations
for elements of the equivalence group of the class~\eqref{EqvcKdVlikeSuperclass}.
In other words, one can assume from the very beginning
that the transformation~$\mathcal T$ has the form~\eqref{EqEquivtransOfvcKdVlikeSuperclass}.
We substitute the expressions for the new variables (with tildes)
into~\eqref{Eqvc5KdV0tilde} and obtain an equation in the old variables (without tildes).
This should be an identity on the manifold defined by~\eqref{Eqvc5KdV0}
in the fifth-order jet space~$J^5$ with the independent variables $(t,x)$ and the dependent variable~$u$.
In order to take into account the constraint between variables of~$J^5$ on the manifold defined by~\eqref{Eqvc5KdV0},
we substitute the expression for~$u_t$ implied by~\eqref{Eqvc5KdV0}.
The splitting of the above identity with respect to~$u$ and its derivatives $u_x$, $u_{xx}$, $u_{xxx}$, $u_{xxxx}$ and $u_{xxxxx}$
results in the determining equations for the functions~$T$, $X^0$, $X^1$, $U^0$ and $U^1$, which include the equation $U^1_x=0$.
If the condition $a^2+b^2+c^2\neq0$ holds, then we additionally get $U^0_x=0$.
The rest of the determining equations is exhausted by
\begin{gather*}
\tilde m U^0 T_t+U^0_t=0,\quad
\tilde aU^1T_t=a (X^1)^3,\quad
\tilde bU^1T_t=b (X^1)^3,\quad
\tilde c(U^1)^2T_t=c X^1,\\
\tilde gT_t=g (X^1)^5,\quad
\tilde fU^1T_t=fX^1-2\tilde cU^1U^0T_t,\quad
\tilde hT_t=h (X^1)^3-\tilde aU^0T_t,\\
\tilde mU^1T_t=mU^1-U^1_t,\quad
\tilde n T_t=nX^1+X^0_t-\tilde kX^0T_t-\tilde c(U^0)^2T_t-\tilde fU^0T_t,\\
\tilde kX^1T_t=kX^1+X^1_t.
\end{gather*}
Introducing the notations $T=\alpha$, $X^1=\beta$, $X^0=\gamma$, and $U^1=\varphi$ and solving the determining equations,
we get the statement of Theorem~1.
\end{proof}

\begin{corollary}
The subclass of the class~\eqref{Eqvc5KdV0} with $fh=0$ and $(a,c)\ne(0,0)$
is normalized with respect to its usual equivalence group
consisting of the transformations~\eqref{Eqvc5KdVEquivGroup} with $\sigma=0.$
\end{corollary}
\medskip
\begin{corollary}
The subclass of the class~\eqref{Eqvc5KdV0} with $k=0$
is normalized with respect to its generalized extended equivalence group
that comprises the transformations~\eqref{Eqvc5KdVEquivGroup} with $\beta={\rm const}$.
\end{corollary}
\medskip
\begin{corollary}\label{Corollary5KdVEqInessentialArbitraryElements}
Any equation from the class~\eqref{Eqvc5KdV0} can be reduced by the point transformation
\begin{equation}\label{EqTransVC5KdV0}
\tilde t=\int ge^{-5\int\!k\,{\rm d}t}\,{\rm d}t,\qquad\tilde x=e^{-\int\!k\,{\rm d}t}x-\int ne^{-\int\!k\,{\rm d}t}\,{\rm d}t,\qquad \tilde u=e^{\int\!m\,{\rm d}t}u
\end{equation}
to an equation from the same class with $g=1$ and $m=n=k=0$.
The subclass of the class~\eqref{Eqvc5KdV0} singled out by the constraints $g=1$ and $m=n=k=0$ is normalized with respect to
its generalized extended equivalence group~$G^\sim_1$ consisting of the transformations~\eqref{Eqvc5KdVEquivGroup} with
\mbox{$\beta_t=\varphi_t=0$}, $\alpha_t=\beta^5$ and $\gamma_t=\sigma\beta f-\sigma^2\beta c$.
\end{corollary}
\medskip

Transformations from the equivalence group $G^\sim$ have a nice specific structure.
The principal property is that they are fiber-preserving
(i.e., the transformation components corresponding to the independent variables~$t$ and~$x$ depend only on these variables) and, moreover, they are linear in~$u$.
An additional bonus is that the transformation component for~$t$ depends only on~$t$ and the transformation component for~$x$ is linear in~$x$.
Given this, the entire study of equations from the class~\eqref{Eqvc5KdV0} within integrability theory can be implemented up to $G^\sim$-equivalence,
which coincides for this class with general contact (resp.\ point) equivalence
since the class~\eqref{Eqvc5KdV0} is normalized with respect to~$G^\sim$ in both the contact-transformation and point-transformation frameworks.
This includes not only local symmetries, cosymmetries, conservation laws, recursion operators, B\"acklund transformations and exact (in particular, $N$-soliton) solutions,
the study of which up to contact or point equivalence is quite common,
but also the Painlev\'e property, bilinear representations and Lax pairs, for which the linearity of elements of~$G^\sim$ in~$u$ is crucial.
In other words, Corollary~\ref{Corollary5KdVEqInessentialArbitraryElements} implies
that four arbitrary elements of the class~\eqref{Eqvc5KdV0}, $g$, $m$, $n$ and~$k$, are inessential
and can be set to canonical values, 1, 0, 0 and~0, from the very beginning,
which would significantly simplify the entire further consideration.

It is also reasonable to use $G^\sim$-equivalence when integrable cases have already been separated from the other equations of the form~\eqref{Eqvc5KdV0}.
Even if the above gauging of the arbitrary elements~$g$, $m$, $n$ and~$k$ by transformations from $G^\sim$ has been carried out,
it is still possible to play with the equivalence group~$G^\sim_1$ of the subclass in which $g=1$ and $m=n=k=0$.
As the integrability property leads to additional relations between arbitrary elements,
the form of integrable equations from the class~\eqref{Eqvc5KdV0} can be simplified more than the form of a general equation from this class
since subgroups of $G^\sim$ parameterized by constants then spring into action.
More specifically, consider the class of constant-coefficient fifth-order KdV equations of the form
\begin{equation}\label{Eqvc5KdVconst}
u_t+Auu_{xxx}+Bu_ xu_{xx}+Cu^2u_x+u_{xxxxx}=0,
\end{equation}
where $A,$ $B$ and $C$ are nonzero constants.
Up to scale transformations, there exist three inequivalent triples $(A,B,C)$ such that the corresponding equations of the form~\eqref{Eqvc5KdVconst} are integrable.
These are the triples $(10,20,30)$, $(15,15,45)$ and $(10,25,20)$~\cite{MikhailovShabatSokolov1991},
which respectively give
\par\medskip\par\noindent$\bullet$ Lax's fifth-order KdV equation~\cite{Lax1968a}
\begin{equation}\label{lax}
u_t+10 uu_{xxx}+20 u_xu_{xx}+30 u^2u_x+u_{xxxxx}=0;
\end{equation}
$\bullet$ the Sawada--Kotera equation~\cite{Sawada&Kotera1974} (equivalent to the Caudrey--Dodd--Gibbon equation~\cite{caud1976a})
\begin{equation}\label{SawadaKotera}
u_t+15 uu_{xxx}+15 u_xu_{xx}+45 u^2u_x+u_{xxxxx}=0;
\end{equation}
$\bullet$ the Kaup--Kupershmidt equation~\cite{Kaup1980}
\begin{equation}
\label{KaupKupershmidt}
u_t+10 uu_{xxx}+25 u_xu_{xx}+20 u^2u_x+u_{xxxxx}=0.
\end{equation}

\begin{corollary}
The usual equivalence group~$G^{\sim}_{\const}$ of the class~\eqref{Eqvc5KdVconst}
consists of the transformations
\begin{gather*}
\tilde t=\beta^5 t+\delta,
\qquad
\tilde x=\beta x+\gamma,
\qquad
\tilde u=\frac{u}{\beta^2\lambda},
\qquad
\tilde A={\lambda}A,\qquad \tilde B={\lambda}B,\qquad \tilde C ={\lambda^2}C.
\end{gather*}
Here  $\beta$, $\gamma$, $\delta$, and $\lambda$ are arbitrary constants with
$\beta\lambda\neq0$.
This group generates the entire equivalence groupoid $\mathcal G^\sim_{\const}$ of the class~\eqref{Eqvc5KdVconst},
i.e., the class~\eqref{Eqvc5KdVconst} is normalized in the usual sense.
\end{corollary}

\newpage

In view of this assertion it is obvious, e.g., that the Caudrey--Dodd--Gibbon equation, in which $(\tilde A,\tilde B,\tilde C)=(30,30,180)$,
is similar to the Sawada--Kotera equation~\eqref{SawadaKotera}.
The similarity between these equations is realized by the scale transformation $\tilde t=t,$ $\tilde x=x$, $\tilde u=\frac12 u$.

\begin{theorem}An equation from the class~\eqref{Eqvc5KdV0} is similar to a constant-coefficient equation of the form~\eqref{Eqvc5KdVconst} with $ABC\ne0$
if and only if its coefficients satisfy the conditions
\begin{equation}\label{criterion}
\left(\frac{b}{a}\right)_t=\left(\frac{b^2}{c g}\right)_t=0,\quad
\left(\frac{f}{c}\right)_t=-m\frac{f}{c},\quad
\left(\frac{b}{g}\right)_t=(m-2k)\frac{b}{g},\quad af=2ch.
\end{equation}
\end{theorem}

The coefficients of all integrable equations considered in~\cite{Xu2013,Yu&Gao&Sun&Liu2010,Yu&Gao&Sun&Liu2012}
(except the family of equations from~\cite{Xu2013} with coefficients presented in Case II) satisfy conditions~\eqref{criterion}.
Therefore, these equations are similar to constant-coefficient ones.
Namely, the equation~\cite{Yu&Gao&Sun&Liu2012}
\begin{gather}\label{eq_Saw_kot_ext}
u_t+15g\Upsilon uu_{xxx}+15g\Upsilon u_xu_{xx}+45g\Upsilon^2u^2u_x+gu_{xxxxx}+mu+n u_x+kxu_x=0,
\end{gather}
where $\Upsilon=\nu e^{\int(m-2k)\,{\rm d}t}$ and $\nu$ is a nonzero constant, is mapped to the Sawada--Kotera equation~\eqref{SawadaKotera}
by the transformation that differs from~\eqref{EqTransVC5KdV0} in the additional scaling of~$u$ by~$\nu$,
\looseness=-1
\begin{equation}\label{Tr_Saw-Kot}
\tilde t=\int ge^{-5\int\!k\,{\rm d}t}\,{\rm d}t,\qquad\tilde x=e^{-\int\!k\,{\rm d}t}x-\int ne^{-\int\!k\,{\rm d}t}\,{\rm d}t,\qquad \tilde u=\nu e^{\int\!m\,{\rm d}t}u.
\end{equation}
The same transformation maps any equation of the form~\eqref{Eqvc5KdV0} with $f=h=0$ and $a$, $b$, and $c$ given by
\begin{gather*}
a=10g\Upsilon,\quad b=20g\Upsilon,\quad c=30g\Upsilon^2\qquad\mbox{or}\qquad a=10g\Upsilon,\quad b=25g\Upsilon,\quad c=20g\Upsilon^2
\end{gather*}
to the constant-coefficient integrable equations~\eqref{lax} or~\eqref{KaupKupershmidt}, respectively.
Therefore, these two integrable cases were missed in~\cite{Yu&Gao&Sun&Liu2012}.

Another point transformation of the form
\[
\tilde t=\frac1{5\mu_1}\!\int\! ae^{-\int\!m\,{\rm d}t}\,{\rm d}t,\quad
\tilde x=x-\int\!\! \left(n-\frac{\mu_2^2}{\mu_1}ae^{-\int\!m\,{\rm d}t}\!\right)\!{\rm d}t,\quad
\tilde u=\kappa\left(e^{\int\!m\,{\rm d}t}u+\frac{\mu_2}{\mu_1}\right)
\]
with $\kappa=\mu_1/3$ (resp. $\kappa=\mu_1/2$) maps equations from the class~\eqref{Eqvc5KdV0} with $k=0$ and the other coefficients satisfying conditions I (resp.\ III)
to  the Sawada--Kotera equation~\eqref{SawadaKotera} (resp.\ the Kaup--Kupershmidt equation~\eqref{KaupKupershmidt}).

Each equation of the form~\eqref{Eqvc5KdV0} with $k=0$ and the coefficients presented
in Case II is reduced to the variable-coefficient equation
\begin{equation}\label{vcInt}
u_t+10 uu_{xxx}+20 u_xu_{xx}+30 u^2u_x+u_{xxxxx}+\psi(t)(6uu_x+u_{xxx})=0,
\end{equation}
where $\psi(t)=\dfrac{10\mu_1}3\dfrac ha e^{\int\!m\,{\rm d}t}$,
by the transformation
\begin{gather*}
\tilde t=\frac3{10\mu_1}\int\! ae^{-\int\!m\,{\rm d}t}\,{\rm d}t,\quad
\tilde x=x-\int\! n\,{\rm d}t,\quad
\tilde u=\frac{\mu_1}3 e^{\int\!m\,{\rm d}t}u.
\end{gather*}
The equation~\eqref{vcInt} is integrable since it is a ``linear superposition'', with time-dependent coefficients,
of Lax's fifth-order KdV equation~\eqref{lax} and the classical KdV equation $u_t+6uu_x+u_{xxx}=0$,
which are integrable and whose associated evolution vector fields commute~\cite[Theorem~3.1]{Fuchssteiner1993}.

Summarizing, we have explained the appearance of all integrable cases obtained in~\cite{Xu2013,Yu&Gao&Sun&Liu2010,Yu&Gao&Sun&Liu2012}
and have found that two integrable cases were missed in~\cite{Yu&Gao&Sun&Liu2010,Yu&Gao&Sun&Liu2012}.
All of these variable-coefficient integrable equations could be constructed
from known integrable equations using equivalence transformations of the class~\eqref{Eqvc5KdV0}.
Moreover, the possibility of gauging four arbitrary elements of the class~\eqref{Eqvc5KdV0} by equivalence transformations
means that these arbitrary elements are inessential within the study carried out in the above papers,
which is, therefore, unreasonably overcomplicated.

\section{Applications of point transformations for further analysis\\ of integrable variable-coefficient equations}
\label{SectionOnApplicationOfPointTransToLaxPairs}

When the similarity of integrable equations from the class~\eqref{Eqvc5KdV0} to well-known integrable equations is established,
further consideration is needless as all objects related to general integrable equations from the class~\eqref{Eqvc5KdV0}
and their properties can be easily derived from those of the classical similar equations using the similarity.
We demonstrate this derivation only for Lax pairs, for which the special structure of transformations from~$G^\sim$ is essential,
since the same procedure, e.g., for symmetries, conservation laws and exact solutions has already become a routine task.

The Sawada--Kotera equation~\eqref{SawadaKotera} admits the Lax pairs
\begin{gather*}
L={\partial_x}^3+3u\partial_x,\\
P=9{\partial_x}^5+45u{\partial_x}^3+45u_x{\partial_x}^2+15(2u_{xx}+3u^2)\partial_x
\end{gather*}
and
\begin{gather*}
L={\partial_x}^3+3u\partial_x+3u_x,\\
P=9{\partial_x}^5+45u{\partial_x}^3+90u_x{\partial_x}^2+15(5u_{xx}+3u^2)\partial_x+30(u_{xxx}+3uu_x).
\end{gather*}
Carrying out the transformation~\eqref{Tr_Saw-Kot} in the associated spectral problems, $L\psi=\lambda\psi$, $\psi_t=P\psi$,
we derive the corresponding Lax pairs for the variable-coefficient equation~\eqref{eq_Saw_kot_ext},
\begin{gather*}
L=e^{3\int\!k\,{\rm d}t}({\partial_x}^3+3\Upsilon u\partial_x),\\
P=9g{\partial_x}^5+45g\Upsilon u{\partial_x}^3+45g\Upsilon u_x{\partial_x}^2+(30g\Upsilon u_{xx}+45g\Upsilon^2u^2-kx-n)\partial_x
\end{gather*}
and
\begin{gather*}
L=e^{3\int\!k\,{\rm d}t}({\partial_x}^3+3\Upsilon u\partial_x+3\Upsilon u_x),\\
P=9g{\partial_x}^5+45g\Upsilon u{\partial_x}^3+90g\Upsilon u_x{\partial_x}^2+(75g\Upsilon u_{xx}+45g\Upsilon^2u^2-kx-n)\partial_x\\
\phantom{P=}{}+30g\Upsilon u_{xxx}+90g\Upsilon^2uu_x,
\end{gather*}
respectively.
Here and in what follows, we again use the notation $\Upsilon=\nu e^{\int(m-2k)\,{\rm d}t}$ with a nonzero constant~$\nu$.

The Kaup--Kupershmidt equation~\eqref{KaupKupershmidt} admits the Lax pair
\begin{gather*}
L={\partial_x}^3+2u\partial_x+u_x,\\
P=9{\partial_x}^5+30u{\partial_x}^3+45u_x{\partial_x}^2+5(7u_{xx}+4u^2)\partial_x+10(u_{xxx}+2uu_x).
\end{gather*}
Using the transformation~\eqref{Tr_Saw-Kot}, we obtain the corresponding Lax pair for the equation
\begin{gather*}
u_t+10g\Upsilon uu_{xxx}+25g\Upsilon u_xu_{xx}+20g\Upsilon^2u^2u_x+gu_{xxxxx}+mu+n u_x+kxu_x=0,
\end{gather*}
which is of the form
\begin{gather*}
L=e^{3\int k\,{\rm d}t}({\partial_x}^3+2\Upsilon u\partial_x+\Upsilon u_x),\\
P=9g{\partial_x}^5+30g\Upsilon u{\partial_x}^3+45g\Upsilon u_x{\partial_x}^2+(35g\Upsilon u_{xx}+20g\Upsilon^2u^2-kx-n)\partial_x\\
\phantom{P=}{}+10g\Upsilon u_{xxx}+20g\Upsilon^2uu_x.
\end{gather*}

Note that the above Lax pairs, which are constructed using the similarity to well-known constant-coefficient integrable equations,
are still associated with isospectral problems
in contrast to, e.g., Lax pairs constructed in~\cite{Yu&Gao&Sun&Liu2012}
directly for variable-coefficient equations.

\section{Discussion}

Admissible contact (resp.\ point) transformations between equations from classes of variable-coefficient PDEs
considered in the literature on integrability usually possess specific properties
including linearity with respect to the dependent variable and more particular fiber preservation and linearity with respect to independent variables
(e.g., when the transformation component for~$t$ depends only on~$t$ and the transformation component for~$x$ is linear in~$x$).
Moreover, these classes of PDEs are, as a rule, normalized in a certain sense,
i.e., their contact (resp.\ point) equivalence groupoids are generated by the corresponding equivalence groups,
although often different generalizations of the notion of usual equivalence group should be considered in order to achieve the normalization.


As a result, equivalence transformations fit well into the study of integrability of variable-coefficient PDEs,
where they can be used for several purposes:
\begin{itemize}\itemsep=0ex
\item
to search for inessential arbitrary elements of the class of variable-coefficient PDEs under consideration
and to gauge these elements to chosen simple values from the outset;
\item
to establish the similarity of integrable variable-coefficient PDEs, which are separated by another method (e.g., the Painlev\'e test) from the class under consideration,
to well-known (usually, constant-coefficient) integrable equations;
or, more generally, to select canonical representatives in the obtained list of integrable equations;
\item
to check listed integrable cases using the established similarity to previously known integrable equations;
\item
to derive all objects related to a singled-out integrable equation and their properties
from those of a similar well-studied integrable equation;
such objects include, but are not exhausted by, local symmetries, cosymmetries, conservation laws, recursion operators, B\"acklund transformations, exact solutions,
the Painlev\'e expansion, bilinear representations and Lax pairs.
\end{itemize}

All of the above purposes are illustrated in the present paper using the class~\eqref{Eqvc5KdV0} of variable-coefficient fifth-order KdV-like equations.
Out of the nine arbitrary elements of the class~\eqref{Eqvc5KdV0}, each of which is a smooth function of~$t$, only five are really essential.
Neglecting this fact leads to a~needless complication of any study of the class~\eqref{Eqvc5KdV0}, cf.\ \cite{Xu2013,Yu&Gao&Sun&Liu2010,Yu&Gao&Sun&Liu2012}.
A~convenient choice for arbitrary elements to be gauged is given by $g$, $m$, $n$ and~$k$, and natural constraints for them are $g=1$ and $m=n=k=0$.
We have shown that all the integrable equations found in \cite{Xu2013,Yu&Gao&Sun&Liu2010,Yu&Gao&Sun&Liu2012}
are similar with respect to point equivalence transformations to
(i.e., in the terminology of~\cite{Burtsev&Zakharov&Mikhailov1987}, are trivial deformations of)
well-known fifth-order integrable KdV-like equations, namely,
the Sawada--Kotera equation~\eqref{SawadaKotera},
the Kaup--Kupershmidt equation~\eqref{KaupKupershmidt}
and the Fuchssteiner time-dependent linear superposition~\eqref{vcInt}
of Lax's fifth-order KdV equation~\eqref{lax} and the classical KdV equation.
This similarity clearly explains the appearance of the three integrable cases found in~\cite{Xu2013}
and indicates that two integrable cases were definitely missed in~\cite{Yu&Gao&Sun&Liu2010,Yu&Gao&Sun&Liu2012}.
It also allows one to obtain all further results of~\cite{Xu2013,Yu&Gao&Sun&Liu2010,Yu&Gao&Sun&Liu2012} in a much easier way
from analogous results for the canonical counterparts, which are already presented in the literature.
Moreover, in the same way, the results of~\cite{Xu2013,Yu&Gao&Sun&Liu2010,Yu&Gao&Sun&Liu2012} can be enhanced and completed.
We have confined ourselves to deriving Lax pairs by the similarity since this is not too customary in contrast to, e.g., solutions in closed form.
The transformational method gives more possibilities for varying the form of Lax pairs than the direct construction of Lax pairs for variable-coefficient equations.
The leading coefficients of operators of Lax pairs presented in Section~\ref{SectionOnApplicationOfPointTransToLaxPairs} are not equal to the unit in general.
At the same time, owing to properly chosen coefficients, these Lax pairs are associated with isospectral problems as in the case of constant-coefficient equations.
Note that, in general, the study of variable-coefficient integrable equations may lead to essentially nonisospectral problems%
~\cite{Burtsev&Zakharov&Mikhailov1987,Sakhnovich1996,Sakhnovich2008}.

Although point transformations were applied in the field of integrability of variable-coefficient PDEs for a long time
\cite{Grimshaw1979a,Grimshaw1979b,Lugovtsov&Lugovtsov1969},
only some classes of such (1+1)-dimensional equations
(e.g., variable-coefficient third- and fifth-order KdV-like equations and variable-coefficient cubic Schr\"o\-din\-ger equations)
have been systematically studied within the transformational framework, see~\cite{Popovych&Vaneeva2010,Vaneeva2012,Vaneeva2013}, references therein and the present paper.
Therefore, a possible direction for investigation is to extend it to other classes of integrable equations or integrable systems of equations
including the case of more independent variables.

There are two regular methods for the construction of integrable evolution equations that explicitly involve the time variable~$t$
using known integrable equations;
these are the transformational method and the Fuchssteiner linear superposition of integrable evolution equations.
All integrable KdV-like equations discussed in Section~\ref{SectionOnfKdV} of the present paper can be obtained by one of these methods or their combination
from classical constant-coefficient KdV-like integrable equations (the Sawada--Kotera equation~\eqref{SawadaKotera},
the Kaup--Kupershmidt equation~\eqref{KaupKupershmidt}, Lax's fifth-order KdV equation~\eqref{lax} and the KdV equation itself).
An interesting question is whether there exists a~KdV-like equation with time-dependent coefficients that cannot be obtained in the above way.
What about other hierarchies?
Another open question concerns the existence of methods
that differ from the transformational and Fuchssteiner methods
but also allow for easy generation of new variable-coefficient integrable equations from known constant-coefficient integrable \mbox{equations.}

\subsection*{Acknowledgements}
The authors are grateful to Profs. A.~Sakhnovich and A.~Sergyeyev for useful discussions. We  also appreciate the valuable comments by the referees and Prof. S.~Lidstr\"om, which have led to essential improvement of the paper.
The research of ROP was supported by the Austrian Science Fund (FWF), project P25064.


\begin{thebibliography}{99}\itemsep=0ex
\footnotesize

\bibitem{Ablowitz&Segur1981}
M.J. Ablowitz and H. Segur,
{\it Solitons and the inverse scattering transform},
SIAM, Philadelphia, 1981.

\bibitem{Bihlo&Cardoso-Bihlo&Popovych2012}
A. Bihlo, E. Dos Santos Cardoso-Bihlo and R.O. Popovych,
Complete group classification of a class of nonlinear wave equations,
{\it J. Math. Phys.} {\bf 53} (2012), 123515, 32 pp.; arXiv:1106.4801.

\bibitem{Bluman&Kumei1980}
G. Bluman and S. Kumei,
On the remarkable nonlinear diffusion equation $(\p/\p x)[a(u+b)^{-2}(\p u/\p x)]-\p u/\p t=0$,
{\it J. Math. Phys.}  {\bf 21} (1980), 1019--1023.

\bibitem{Burgers1948}
J.M. Burgers,
A mathematical model illustrating the theory of turbulence,
{\it Adv. Appl. Mech.} {\bf 1} (1948), 171--199.

\bibitem{Burtsev&Zakharov&Mikhailov1987}
S.P. Burtsev, V.E. Zakharov and A.V. Mikhailov,
The inverse problem method with a variable spectral parameter,
{\it Teoret. Mat. Fiz.}  {\bf 70}  (1987), 323--341. (in Russian)
English translation: {\it Theoret. and Math. Phys.} {\bf 70} (1987), 227--240.

\bibitem{Calogero1991}
F. Calogero,
Why are certain nonlinear PDEs both  widely applicable and integrable?,
in: V.E. Zakharov (Ed.), {\it What is integrability?},  Springer Ser. Nonlinear Dynam., Springer-Verlag, Berlin, 1991, pp.~1--62.

\bibitem{caud1976a}
P.J. Caudrey, R.K. Dodd and J.D. Gibbon,
A new hierarchy of Korteweg--de Vries equations,
{\it Proc. R. Soc. Lond. A.} {\bf 351} (1976), 407--422.

\bibitem{Cole1951}
J.D. Cole,
On a quasilinear parabolic equation occurring in aerodynamics,
{\it Quart. Appl. Math.} {\bf 9} (1951), 225--236.


\bibitem{Fokas&Yortsos1982}
A.S. Fokas, Y.C.  Yortsos,
On the exactly solvable equation $S_t=\lbrack(\beta S+\gamma)^{-2}S_x\rbrack_x+\alpha(\beta S+\gamma)^{-2}S_x$ occurring in two-phase flow in porous media,
{\it SIAM J. Appl. Math.}  {\bf 42} (1982), 318--332.

\bibitem{Fuchssteiner1993}
B. Fuchssteiner, Integrable nonlinear evolution equations with
time-dependent coefficients,
{\it J. Math. Phys.}  {\bf 34} (1993), 5140--5158.

\bibitem{Gardner&Greene&Kruskal&Miura1967}
C.S. Gardner, J.M. Greene, M.D. Kruskal and R.M. Miura,
Method for solving the Korteweg--de Vries equation,
{\it Phys. Rev. Let.} {\bf 19} (1967), 1095--1097.

\bibitem{Grimshaw1979a}
R. Grimshaw,
Slowly varying solitary waves. I. Korteweg--de Vries equation,
{\it Proc. R. Soc. Lond. A} {\bf 368} (1979), 359--375.

\bibitem{Grimshaw1979b}
R. Grimshaw,
Slowly varying solitary waves. II. Nonlinear Schr\"odinger equation,
{\it Proc. R. Soc. Lond. A} {\bf 368} (1979), 377--388.

\bibitem{HammertonCrighton1989}
P.W. Hammerton and D.G. Crighton,
Approximate solution methods for nonlinear acoustic propagatioin over long ranges,
{\it Proc. R. Soc. Lond. A} {\bf 426} (1989), 125--152.

\bibitem{Hirota1971}
R. Hirota,
Exact solution of the Korteweg--de Vries equation for multiple collisions of solitons,
{\it Phys. Rev. Lett.} {\bf 27} (1971), 1192--1194.

\bibitem{Hopf50}
E. Hopf,
The partial differential equation $u_t+uu_x=\mu u_{xx}$,
{\it Comm. Pure Appl. Math.} {\bf 33}  (1950), 201--230.

\bibitem{ips2010a}
N.M. Ivanova, R.O. Popovych and C. Sophocleous,
Group analysis of variable coefficient diffusion-convection equations. I. Enhanced group classification,
{\it Lobachevskii J. Math.} {\bf 31} (2010), 100--122; arXiv:0710.2731.

\bibitem{Johnson1979}
R.S. Johnson,
On the inverse scattering transform, the cylindrical Korteweg--de Vries equation and similarity solutions,
{\it Phys. Lett. A} {\bf 72} (1979), 197--199.

\bibitem{Joshi1987}
N. Joshi,
Painlev\'e property of general variable-coefficient versions of the Korteweg--de Vries and nonlinear Schr\"odinger equations,
{\it Phys. Lett. A} {\bf 125} (1987), 456--460.

\bibitem{Kaup1980}
D.J. Kaup, On the inverse scattering problem for cubic eigenvalue problems of the class $\psi \sb{xxx}+6Q\psi \sb{x}+6R\psi =\lambda \psi $,
{\it Stud. Appl. Math.} {\bf 62} (1980), 189--216.

\bibitem{Korteweg&de_Vries1895}
D.J. Korteweg and G. de Vries,
On the change of form of long waves advancing in a rectangular canal, and on a new type of long stationary waves,
{\it Phil. Mag.} {\bf 39} (1895), 422--443.

\bibitem{Kingston&Sophocleous1998}
J.G. Kingston and C. Sophocleous, On form-preserving point transformations of partial differential equations,
{\it J. Phys. A: Math. Gen.} {\bf 31} (1998), 1597--1619.

\bibitem{Kruskal1975}
M. Kruskal, Nonlinear wave equations,
in: J. Moser (Ed.), {\it  Dynamical systems, theory and applications (Rencontres, Battelle Res. Inst., Seattle, Wash., 1974)},
Lecture Notes in Phys., Vol. 38, Springer, Berlin, 1975, pp. 310--354.

\bibitem{Lax1968a}
P.D. Lax,
Integrals of nonlinear equations of evolution and solitary waves,
{\it Comm. Pure Appl. Math.} {\bf 21} (1968), 467--490.

\bibitem{Lugovtsov&Lugovtsov1969}
A. Lugovtsov and B. Lugovtsov,
Axially symmetric long waves in the KdV approximation,
{\it Dynamika sploshnoi sredy} {\bf 1} (1969), 195--200. (in Russian)

\bibitem{Magadeev1993}
B.A. Magadeev, On group classification of nonlinear evolution equations,
{\it Algebra i Analiz} {\bf 5} (1993), 141--156 (in Russian);
English translation in {\it St. Petersburg Math. J.} {\bf 5} (1994), 345--359.

\bibitem{Matveev&Salle1991}
V.B. Matveev and M.A. Salle
{\it Darboux transformations and solitons},
Springer-Verlag, Berlin, 1991.

\bibitem{MikhailovShabatSokolov1991}
A.V. Mikhailov, A.B. Shabat and V.V. Sokolov,
The symmetry approach to classification of integrable equations,
in: V.E. Zakharov (Ed.), {\it What is interability?},  Springer Ser. Nonlinear Dynam., Springer-Verlag, Berlin, 1991, pp.~115--184.

\bibitem{Miura1968}
R.M. Miura,
Kortweg--de Vries equation and generalizations. I. A remarkable explicit nonlinear transformation,
{\it J. Math. Phys.} {\bf 9} (1968), 1202--1204.

\bibitem{Miura&Gardner&Kruskal1968}
R.M. Miura, C.S. Gardner and M.D. Kruskal,
Korteweg--de Vries equation and generalizations. II. Existence of conservation laws and constants of motion,
{\it J. Math. Phys.} {\bf 9} (1968), 1204--1209.

\bibitem{Olver1977}
P.J. Olver, Evolution equations possessing infinitely many symmetries,
{\it J. Math. Phys.} {\bf 18} (1977), 1212--1215.

\bibitem{Olver1993}
P. Olver, {\it Applications of Lie groups to differential equations},
Springer-Verlag, New York, 1993.

\bibitem{Ovsiannikov1982}
L.V. Ovsiannikov, {\it Group analysis of differential equations},
Academic Press, New York, 1982.

\bibitem{Popovych&Bihlo2012}
R.O. Popovych and A. Bihlo,
Symmetry preserving parameterization schemes,
{\it J. Math. Phys.} {\bf 53} (2012), 073102, 36 pp.

\bibitem{popo2010a}
R.O. Popovych, M. Kunzinger and H. Eshraghi,
Admissible transformations and normalized classes of nonlinear Schr\"odinger equations,
{\it Acta Appl. Math.} {\bf  109} (2010), 315--359, arXiv:math-ph/0611061.

\bibitem{Popovych&Vaneeva2010}
R.O. Popovych and O.O. Vaneeva,
More common errors in finding exact solutions of nonlinear differential equations: Part I,
{\it Commun. Nonlinear Sci. Numer. Simulat.} {\bf 15} (2010), 3887--3899; arXiv:0911.1848.

\bibitem{Rogers&Schief2002}
C. Rogers and W.K. Schief,
{\it B\"acklund and Darboux transformations. Geometry and modern applications in soliton theory},
Cambridge University Press, Cambridge, 2002.

\bibitem{Sakhnovich1996}
A. Sakhnovich,
Iterated B\"acklund--Darboux transformation and transfer matrix-function (nonisospectral case),
{\it Chaos Solitons Fractals} {\bf 7}  (1996), 1251--1259.

\bibitem{Sakhnovich2007}
A. Sakhnovich,
B\"acklund--Darboux transformation for non-isospectral canonical system and Riemann--Hilbert problem,
{\it SIGMA} {\bf 3} (2007), Paper 054, 11 pp.

\bibitem{Sakhnovich2008}
A. Sakhnovich,
Nonisospectral integrable nonlinear equations with external potentials and their GBDT solutions,
{\it J. Phys. A: Math. Theor.} {\bf 41} (2008),  155204, 15 pp.

\bibitem{Sawada&Kotera1974}
K. Sawada and T. Kotera,
A method for finding $N$-soliton solutions of the K.d.V. equation and K.d.V.-like equation,
{\it Progr. Theoret. Phys.} {\bf 51} (1974), 1355--1367.

\bibitem{encycl}
A.B. Shabat, V.E. Adler, V.G. Marikhin, V.V. Sokolov (Eds.),
{\it Encyclopedia of integrable systems},
L.D.~Landau Institute for Theoretical Physics, 2010 (http://home.itp.ac.ru/${}^\sim$adler/E/e.pdf).

\bibitem{Sharma&Tasso1977}
A.S. Sharma and H. Tasso,
{\it Connection between wave envelope and explicit solution of a nonlinear dispersive wave equation},
Report IPP 6/158, 1977.

\bibitem{Storm1951}
M.L. Storm,
Heat conduction in simple metals,
{\it J. Appl. Phys.} {\bf 22} (1951), 940--951.

\bibitem{Strampp1982b}
W. Strampp,
B\"acklund transformations for diffusion equations,
{\it Phys. D} {\bf 6} (1982),  113--118.

\bibitem{Vaneeva2012}
O.O. Vaneeva,
Lie symmetries and exact solutions of variable coefficient mKdV equations: an equivalence based approach,
{\it Commun. Nonlinear Sci. Numer. Simul.} {\bf 17} (2012), 611--618; arXiv:1104.1981.

\bibitem{Vaneeva2013}
O.O. Vaneeva,
Group classification of variable coefficient KdV-like equations,
in: V.~Dobrev (ed.), {\it Springer Proceedings in Mathematics {\rm \&} Statistics, Vol. 36. IX International Workshop ``Lie Theory and its Application in Physics''},
Springer, 2013, pp.~451--459; arXiv:1204.4875.

\bibitem{vps2009}
O.O. Vaneeva, R.O. Popovych and C. Sophocleous,
Enhanced group analysis and exact solutions of variable coefficient semilinear diffusion equations with a power source,
{\it Acta Appl. Math.} {\bf 106} (2009), 1--46; arXiv:0708.3457.

\bibitem{vane2012b}
O.O. Vaneeva, R.O. Popovych and C. Sophocleous,
Extended group analysis of variable coefficient reaction--diffusion
equations with exponential nonlinearities,
{\it J. Math. Anal. Appl.} {\bf 396} (2012), 225--242; arXiv:1111.5198.

\bibitem{Winternitz92}
P. Winternitz and J.P. Gazeau,
Allowed transformations and symmetry classes of variable coefficient Korteweg--de Vries equations,
{\it Phys. Lett. A} {\bf 167} (1992), 246--250.

\bibitem{Xu2013}
G.-Q. Xu,
Painlev\'e integrability of a generalized fifth-order KdV equation with
variable coefficients: Exact solutions and their interactions, {\it Chin. Phys. B} {\bf 22} (2013), 050203, 8~pp.

\bibitem{Yu&Gao&Sun&Liu2010}
X. Yu, Y.-T. Gao, Z.-Y. Sun and Y. Liu,
$N$-soliton solutions, B\"acklund transformation and Lax pair for a generalized variable-coefficient fifth-order Korteweg--de Vries equation,
{\it Phys. Scr.} {\bf 81} (2010), 045402, 6~pp.

\bibitem{Yu&Gao&Sun&Liu2012}
X. Yu, Y.-T. Gao, Z.-Y. Sun and Y. Liu,
Investigation on a nonisospectral fifth-order Korteweg--de Vries equation generalized from fluids,
{\it J. Math. Phys.} {\bf 53} (2012), 013502, 8~pp.

\bibitem{Zabusky&Kruskal1965}
N.J. Zabusky  and M.D. Kruskal,
Interaction of ``solitons'' in a collisionless plasma and the recurrence of initial states,
{\it Phys. Rev. Lett.} {\bf 15} (1965), 240--243.

\end{thebibliography}
\end{document}